\documentclass{article}
\usepackage[utf8]{inputenc}
\usepackage{tikz}
\usepackage{graphicx}
\usepackage{geometry}  
\usepackage{thm-restate}
\usepackage{amsthm}

\newtheorem*{example}{Example}
\usetikzlibrary{decorations.pathreplacing,backgrounds}

\usepackage{array}
\usepackage{wrapfig}
\usepackage{xcolor}
\usepackage{amssymb}
\usepackage{longtable}
\usepackage{amsthm}
\usepackage{amsmath}
\usepackage{algorithmicx}
\usepackage{algorithm}
\usepackage[noend]{algpseudocode}
\usepackage{adjustbox}

\usepackage{physics}
\usepackage{multirow}
\usepackage{mathtools}
\usepackage{tikz}
\usepackage{tablefootnote}

\usetikzlibrary{decorations.pathreplacing}
\usetikzlibrary{cd}
\usetikzlibrary{positioning}
\usetikzlibrary{shapes,patterns}
\usepackage{listings}
\usepackage{tabu}
\usepackage[margin=.2cm]{caption}
\definecolor{DarkGreen}{rgb}{0.1,0.5,0.1}
\definecolor{DarkRed}{rgb}{0.5,0.1,0.1}
\definecolor{DarkBlue}{rgb}{0.1,0.1,0.5}
\usepackage{cleveref}

\newtheorem{thm}{Theorem}[section]
\newtheorem*{thm*}{Theorem}
\newtheorem{prop}[thm]{Proposition}
\newtheorem{lem}[thm]{Lemma}
\newtheorem{remark}[thm]{Remark}

\newtheorem{claim}[thm]{Claim}
\newtheorem{defi}[thm]{Definition}

\newtheorem{obs}[thm]{Observation}

\renewcommand{\epsilon}{\varepsilon}

\newcommand{\F}{\mathbb{F}}

\newcommand{\cC}{\mathcal{C}}

\newcommand{\Fp}{\F_p}

\newcommand{\ep}[1]{\boldsymbol{e}_p \left( #1 \right)}
\newcommand{\ceil}[1]{\left \lceil #1 \right \rceil}
\newcommand{\ronicomment}[1]{{\color{red}[Roni: #1]}}

\newcommand{\cL}{\mathcal{L}}
\newcommand\todo[1]{}

\newcommand{\njs}[1]{}
\newcommand{\ph}[1]{}

\def\multiset#1#2{\ensuremath{\left(\kern-.3em\left(\genfrac{}{}{0pt}{}{#1}{#2}\right)\kern-.3em\right)}}

\newcommand\blfootnote[1]{%
  \begingroup
  \renewcommand\thefootnote{}\footnote{#1}%
  \addtocounter{footnote}{-1}%
  \endgroup
}
\newcommand{\Leak}{\textup{Leak}_{\tau}}
\begin{document}
\newcommand{\zachicomment}[1]{{\color{brown}[Zachi: #1]}}
\newcommand{\noahcomment}[1]{{\color{orange}[Noah: #1]}}

\title{Repairing Reed-Solomon Codes over Prime Fields via Exponential Sums}

\author{
    Roni Con\thanks{Department of Computer Science, Tel Aviv University. Email: roni.con93@gmail.com} \and 
    Noah Shutty \thanks{Stanford Institute for Theoretical Physics, Stanford University, noaj@alumni.stanford.edu} 
    \and Itzhak Tamo \thanks{Department of Electrical Engineering-Systems, Tel Aviv University. Email: tamo@tauex.tau.ac.il}
    \and Mary Wootters \thanks{Department of Computer Science and Electrical Engineering, Stanford University, Email: marykw@stanford.edu}
   	}
\maketitle
\begin{abstract}
    This paper presents two repair schemes for low-rate Reed-Solomon (RS) codes over prime fields that can repair any node by downloading a \emph{constant} number of bits from each surviving node.
The total bandwidth resulting from these schemes is greater than that incurred during trivial repair; however, this is particularly relevant in the context of \emph{leakage-resilient secret sharing}. In that framework, our results provide attacks showing that $k$-out-of-$n$ Shamir's Secret Sharing over prime fields for small $k$ is \emph{not} leakage-resilient, even when the parties leak only a \emph{constant} number of bits. To the best of our knowledge, these are the first such attacks.

Our results are derived from a novel connection between exponential sums and the repair of RS codes. Specifically, we establish that non-trivial bounds on certain exponential sums imply the existence of explicit nonlinear repair schemes for RS codes over prime fields.
   \blfootnote{
    The work of Itzhak Tamo and Roni Con was supported by the European Research Council (ERC grant number 852953). Noah Shutty was supported in part by NSF DGE-1656518.  Mary Wootters was supported in part by NSF grants CCF-2133154 and CCF-1844628.
    }
     \blfootnote{A preliminary version of this work was presented at the IEEE International Symposium on Information Theory (ISIT) 2023.}

\end{abstract}

\newpage


 
\section{Introduction}

Reed-Solomon (RS) codes are a widely-used family of codes, both in theory and in practice.
Among their many applications, RS codes are utilized in distributed storage systems, as evidenced in systems like Facebook, IBM, Google, etc. (refer to Table 1 in \cite{dinh2022practical}). 
In such systems, a large file is encoded using an erasure-correcting code and then distributed across multiple nodes. Upon the failure of a node, it is desirable to efficiently set up a replacement node using information from the remaining nodes. In this work, we focus on the \emph{repair bandwidth}---the total amount of information downloaded---as our metric of efficiency. The challenge of recovering a failed node with minimal repair bandwidth was first addressed in the seminal paper by Dimakis et al. \cite{dimakis2010network} and has since been a significant topic of research.


 
We begin by defining RS codes.
\begin{defi}
		Let $\alpha_1, \alpha_2, \ldots, \alpha_n$ be distinct points of the finite field $\mathbb{F}_q$ of order  $q$. For $k<n$ the $[n,k]_q$ \emph{RS code} 
		defined  by the  evaluation set $\{ \alpha_1, \ldots, \alpha_n \}$ is the   set of codewords 
		\[
		\left \lbrace \left( f(\alpha_1), \ldots, f(\alpha_n) \right) \mid f\in \mathbb{F}_q[x],\deg f < k \right \rbrace \;.
		\]
		When $n = q$, the resulting  code is called a \emph{full-length RS code}.
\end{defi}

When the code used in the storage system is an RS code, the repair problem can be viewed as a variation of the standard polynomial interpolation problem: Repairing, say, the \(i\)th node, is equivalent to recovering an evaluation \(f(\alpha_i)\) using minimal information from the evaluations \(f(\alpha_j)\) for \(j \neq i\).
Formally, a \emph{repair scheme} for an \([n,k]_q\) RS code with evaluation points \(\alpha_1, \ldots, \alpha_n\) consists of functions \(\tau_j: \mathbb{F}_q  \to \{0,1\}^m\) for each \(j \neq i\); and a repair function \(G: \{0,1\}^{m \times (n-1)} \times [n] \to \mathbb{F}_q\), for some parameter \(m\). We note that the function $\tau_j, G$ may depend on the identity of the failed node $i$ and should perhaps be called $\tau_{j,i}$ and $G_i$ respectively; however, for the same of notational clarity, we omit the $i$ subscript.

For a failed node \(i\in [n]\), and any polynomial \(f\) of degree less than \(k\) (representing the stored data), each surviving node \(j \neq i\) sends a message \(\tau_j(f(\alpha_j))\). The repair scheme then takes in the messages \(\tau_j(f(\alpha_j))\), along with \(i\), and outputs the missing information \(f(\alpha_i)\) using \(G\). 
Our goal is to minimize \(m\), the number of bits sent by each node.\footnote{We note that traditionally, the definition of a repair scheme permits contacting only a subset of the nodes, with each node potentially sending varying amounts of information. In the definition provided above, we assume the scenario where all surviving nodes transmit the same amount of information. Later, we will relax this definition to allow for some of the surviving nodes to abstain from sending any information.}

Via standard polynomial interpolation, it is clear that any $k$ values of $f(\alpha_j)$ suffice to recover the polynomial $f$, and in particular, to recover  $f(\alpha_i)$.  This requires $k \log_2(q)$ bits of information; we refer to this as \emph{trivial repair}. 
However, as was shown in a line of work including \cite{shanmugam2014repair,guruswami2017repairing,tamo2017optimal,con2021nonlinear}, it is possible to do better!  That is, it is possible to recover $f(\alpha_i)$ using strictly less than $k \log(q)$ bits from the other nodes!


\subsection{Repairing Reed-Solomon Codes over Prime Fields} All of the known RS repair schemes but the one in \cite{con2021nonlinear} require that the underlying field be an extension field. In contrast, our work focuses on \emph{prime fields}. Our main motivation stems from applications in \emph{secret sharing}.

In secret sharing, the well-known \emph{Shamir's secret sharing scheme} (Shamir's SSS) is analogous to RS codes. Informally speaking (see \Cref{sec:related} for more details), repair schemes for RS codes in which every node transmits a small number of bits are equivalent to attacks on Shamir's SSSs when the parties may each leak a small amount of adversarially selected information about their shares. In essence, a repair scheme for an RS code implies that an instance of Shamir's SSS is not \emph{leakage-resilient}.
 
 Since the schemes of \cite{shanmugam2014repair, guruswami2017repairing, tamo2017optimal} and others require extension fields, a natural hope for constructing Shamir's SSSs that are leakage-resilient is to work over prime fields, and indeed several works~\cite{benhamouda2021local,nielsen2020lower,maji2021leakage,maji2022improved,klein2023new} 
 have shown that, when the dimension $k$ of the code is large, $\Omega(n)$, Shamir's SSS over prime fields is leakage resilient.  In particular, their results imply that for any constant $m$, there is some constant $\alpha \in (0,1)$ so that any $[n,k]_p$ RS code over a prime field $\mathbb{F}_p$ with $k \geq \alpha n$ does \emph{not} admit a repair scheme that downloads $m$ bits from each surviving node (see Theorem~\ref{thm:secret-sharing-leakage-cons-rate}). 

On the other hand, \cite{con2021nonlinear} showed the existence of asymptotically optimal repair schemes for RS codes of dimension \(k=2\) over sufficiently large prime fields $\mathbb{F}_p$. In their work, each party leaks a non-constant number of bits, specifically, \(m = \frac{1}{n-2} \log(p) + O_n(1)\) bits (see \Cref{tab:results-table} for their results for larger \(k\)).
Our results continue the line of work initiated by \cite{con2021nonlinear}. We make progress by presenting schemes where each node transmits a \emph{constant} number of bits, that is, $m = O(1)$; equivalently, for Shamir's SSS, each party leaks a constant number of bits.
Therefore, our results have implications both for attacks on Shamir's SSS and for distributed storage systems.

\begin{table}
		\begin{center}
			\begin{tabular}{|| c | c | c| p{4cm}||}
				\hline
				& RS Code & Bandwidth per node & Remarks \\ [0.5ex]
				\hline
				\cite[Theorem 3.3]{con2021nonlinear} & $[n,2]_p$ & $\frac{\log(p)}{n-2} +O_n(1) $ & Non-explicit construction.\\
                    \hline
				\cite[Theorem 4.3]{con2021nonlinear} & $[n,k]_p$ & $\frac{\log(p)}{n-k} +O_n(1) $ & Repair only possible for a particular failed node (rather than any failed node)\tablefootnote{This is achieved by puncturing the code in \cite[Theorem 4.3]{con2021nonlinear}.} using all the remaining nodes. \\
                \hline
				 \Cref{thm:shortKloostermanRepair} & $[p^{\delta},3]_p$ & $3$ & $\frac{(\ln \ln p)^\frac{1}{3}}{(\ln p)^\frac{2}{3}}\leq \delta\leq \frac{1}{2}$\\
                \hline
				 \Cref{thm:Weil-recover-const} & $[p,\Theta_B(\sqrt{p})]_p$ & $B\geq 3$ & Recovers the entire polynomial\\
				[1ex]\hline
			\end{tabular}
		\end{center}
		\caption{Our result compared to \cite{con2021nonlinear}. If not stated in the Remarks column, all of the results presented are explicit constructions, and all the results give repair schemes that repair any single failed node using all the remaining nodes.  For all results,  $p$ is a sufficiently large prime. For the results in \cite{con2021nonlinear}, $n$ and $k$ are assumed to be constants relative to $p$.}
		\label{tab:results-table}
	\end{table}

\subsection{Our Results}
In this paper, we present two schemes where every node of an RS code transmits a \emph{constant} number of bits. Both schemes are explicit, meaning that the functions \(\tau_j\) computed by the nodes as well as the reconstruction function \(G\) are explicitly defined. It is noteworthy that the first scheme is a repair scheme, while the second scheme not only enables the repair of a failed node but is, in fact, a decoding scheme that facilitates the recovery of the original codeword.  

This latter point might raise some red flags for the reader familiar with regenerating codes, as it implies that the scheme downloads enough information to recover the entire codeword, while the goal in regenerating codes is to generally download less information than that.  In fact, both of our schemes download \emph{more} information than  $k \cdot \log(p)$ bits, the amount required to recover the entire codeword.  

The fact that our schemes have such high download costs means that they are not competitive with the ``trivial'' repair scheme (which downloads $\log(p)$ bits from each of $k$ nodes) if one cares only about total bandwidth.
However, our schemes are useful in settings where each node can only transmit very few bits; in such a setting, the trivial repair is not an option.  This could be the case in distributed storage models (e.g., in a model where each link can only transmit a few bits).  It is also the case in the model of leakage-resilient secret sharing.
In particular, our results give attacks on Shamir's SSSs over prime fields, where each leaking party need only send a constant number of bits. To the best of our knowledge, these are the first such attacks with a constant number of bits per leaking party.

With that in mind, we present our two main results below.
Throughout this paper, let \( p \) be a prime number. We proceed to present our first result: a repair scheme for an RS code of dimension~$3$).

\begin{thm*}
    [informal, see \Cref{thm:shortKloostermanRepair}]
    For $\exp((\ln p)^{2/3} (\ln \ln p)^{1/3})\leq n \leq \sqrt{p}$, there exists an $[n,3]_p$ RS code where any node can be repaired by downloading \emph{three} bits from each of the $n-1$ remaining nodes. 
\end{thm*}

Our second result is a decoding scheme for a full-length RS code, which enables the recovery of the original codeword using a constant number of bits from a subset of the nodes. This implies that not all nodes are required to participate in the decoding process. We prove the following.
\begin{thm*} [informal, see \Cref{thm:Weil-recover-const}]
    Let $B\geq 3$ be an integer. The full length $[p,k]_p$ RS code admits a decoding scheme by downloading $B$ bits from any $p-m$ nodes where $k+m \leq \cos \left({\frac{2\pi}{2^B} + \frac{2\pi}{p}} \right) \sqrt{p}$. 
\end{thm*}

\subsection{Related Work}\label{sec:related}
\paragraph{Low-bandwidth repair for distributed storage.}
As mentioned above, the low-bandwidth repair problem was introduced in \cite{dimakis2010network}, and since then there have been many code constructions and repair schemes aiming at optimal repair, for example~\cite{el2010fractional,goparaju2013data,papailiopoulos2013repair,tamo2012zigzag,wang2016explicit,rashmi2011optimal,ye2017explicithighrate,ye2017explicit,goparaju2017minimum}. 
The study of repairing Reed-Solomon codes was introduced in~\cite{shanmugam2014repair}, and the works \cite{guruswami2017repairing,tamo2017optimal}, among others, have constructed repair schemes for RS codes over extension fields with total bandwidth that is much smaller than the trivial bandwidth; in particular, the work \cite{tamo2017optimal} showed how to construct RS codes that achieve the \emph{cut-set bound} proved in \cite{dimakis2010network}, i.e., with the minimal possible bandwidth.

Our work focuses on prime fields.  To the best of our knowledge, the only work that gives positive results for repairing RS codes over prime fields is~\cite{con2021nonlinear}, discussed above and summarized in \Cref{tab:results-table}. 

\paragraph{Leakage-Resilient Secret Sharing.}
 Shamir's SSS, which was first introduced in \cite{shamir1979share}, is a fundamental cryptographic primitive that provides a secure method to   distribute a secret among   different parties, so that any $k$ can recover the secret but no $k-1$ learn anything about it.  
Formally, given a secret $s \in \mathbb{F}$, a \emph{reconstruction threshold} $k > 0$ and a number of parties $n$, Shamir's SSS works as follows.  A dealer chooses a random polynomial $f$ of degree at most $k-1$, so that $f(0) = s$.  Then party $i$ is given the share $f(\alpha_i)$, where $\alpha_1, \ldots, \alpha_n \in \mathbb{F}$ are distinct, pre-selected points.  It is easy to verify that any $k$ parties can reconstruct the polynomial $f(x)$ and therefore recover the secret $f(0)=s$, while the shares of any $k-1$ parties reveal no information about the secret.  

Benhamouda, Degwekar,   Ishai,  and  Rabin
 \cite{benhamouda2021local} considered the question  of \emph{local leakage-resilience} of secret sharing schemes over prime fields and, in particular, Shamir's SSS over prime fields. 
In this setting, the model is that each party $i$ may adversarially leak some function $\tau_i(f(\alpha_i)) \in \{0,1\}^m$ of their share, for some (small) $m$.  
A scheme is $m$-local leakage resilient if for any two secrets $s,s'$, the total variation distance between the leaked messages under $s$ and the leaked messages under $s'$ is negligible.
The work of Benhamouda et al. established the following. 
\begin{thm} \label{thm:secret-sharing-leakage-cons-rate} \cite[Corollary 4.12]{benhamouda2021local}
	Let $m$ be a constant positive integer and let $n$ go to infinity. There exists an $\alpha_m<1$, for which Shamir's SSS with $n$ players and threshold $k = \alpha_m n$ is $m$-local leakage resilient. 
\end{thm}
We also note that Benhamouda et al. \cite{benhamouda2021local} phrased a conjecture that for large enough $n$, Shamir's SSS with threshold $\alpha n$ is $1$-bit local leakage resilient for any constant $\alpha>0$. Despite considerable progress \cite{maji2021leakage,maji2022improved,klein2023new}, this conjecture is still open.
%
%

Given \Cref{thm:secret-sharing-leakage-cons-rate}, it is not feasible to repair RS codes over prime fields with a rate arbitrarily close to \(1\) while only leaking a constant number of bits from each party.
In \cite[Lemma 2]{nielsen2020lower}, it is shown that Shamir's SSS are not leakage resilient for sublinear threshold. Specifically,
\begin{lem} \cite[Lemma 2]{nielsen2020lower}
    Let $p$ be the smallest prime larger than $n$. Then, for any constant $c\in (0,1)$ and large enough $n$, it holds that Shamir's SSS over $\Fp$ with $n$ players and threshold $k=cn/\log n$ is not $1$-local leakage resilient.
\end{lem}

However, prior to this work, no explicit local attacks were known. By explicit attacks, we mean explicit small image functions that adversary applies on the shares. 
Our results explicit state such functions and we also show that adversary not only gains significant insight into the secret but, in fact, can completely recover it.

\subsection{Organization} The preliminaries are given Section~\ref{sec:prelim}. 
The repair scheme is given in \Cref{sec:repair-kls} and the decoding scheme is given in \Cref{sec:recover-poly}. In \Cref{sec:summary}, we summarize and present some open questions.

\section{Preliminaries}
    \label{sec:prelim}
    We begin with some needed notations.  For integers $a<b$ let $[a, b]=\{a,a+1, \ldots, b\}$ and $[a]=\{1,2,\ldots, a\}$. An arithmetic progression in some field $\mathbb{F}$ of length $N$ and a step $s\in \mathbb{F}$ is a set of the form $\{a, a+s, \ldots, a+(N-1)s \}$ for some $a \in \mathbb{F}$. 
    Throughout, let $p\neq 2$ be a prime number and $\Fp$ be the finite field of size $p$. We will realize $\Fp$ as the set $[-\frac{p-1}{2},\frac{p-1}{2}]$, where each element of $\Fp$ is viewed as an element of $[-\frac{p-1}{2},\frac{p-1}{2}]$, and  addition and multiplication are computed modulo $ p$.

    For an element $\gamma \in \Fp$ let $\gamma \cdot A:=\{\gamma \cdot a : a\in A \}$ be all the possible products of $\gamma$ with elements in $A$, hence,  the notation $\alpha \in \gamma\cdot [-t, t]$ for  $\alpha, \gamma,t \in \Fp$ implies that there exists  $j\in [-t,t]$ for which $\alpha \equiv \gamma \cdot j \mod p$.  
    Finally, let $\ep{x}:\Fp\rightarrow \mathbb{C}$ be the standard additive character of $\Fp,\ep{x} := \exp\left(\frac{2\pi \sqrt{-1}}{p} \cdot x\right)$.
    
\subsection{Exponential sum bounds}
Exponential sums have been a subject of extensive study in number theory and have found numerous applications in coding theory. In this work, we explore another application of these sums. Specifically, we demonstrate how bounds such as the Weil bound and those on Kloosterman sums can be leveraged to develop repair schemes for Reed-Solomon (RS) codes over prime fields, where each node transmits a constant number of bits. In this section, we introduce several key bounds, beginning with the \emph{Weil bound}.

\begin{thm} \cite[Theorem 5.38]{lidl1997finite} \label{thm:weil-bound}
    Let $f\in \Fp[X]$ be a non-constant polynomial of degree at most $k$. It holds that
    \[
    \left|\sum_{\alpha \in \Fp} \ep{f(\alpha) }\right| \leq  (k-1) \cdot \sqrt{p} \;.
    \]
\end{thm}
As one can see, the bound is non-trivial only when $\deg(f) \leq \sqrt{p}$. The Weil bound has applications throughout mathematics, theoretical computer science, and information theory. Motivated by these applications, there are extensions to the Weil bound that improve the bound in certain classes of polynomials, see e.g., \cite{cochrane2005improved,kaufman2011new}. 
Another important family of exponential sums are \emph{Kloosterman sums}, which first appeared in \cite{kloosterman1927representation} are sums of the form
\[
\left|\sum_{\alpha \in A} \ep{ \frac{a}{\alpha} + b \alpha}\right| \;,
\]
where $A\subseteq \Fp^*$ and $\gcd(a, p) \neq 0$.
Bounds on these sums also have broad applications, mainly in analytical number theory \cite{heath2000arithmetic,iwaniec2021analytic}, but also in coding theory \cite{helleseth1999z4,moisio2007moments,zinoviev2019classical}.


\subsection{Repair Schemes via Arithmetic Progressions} 
\label{sec:framework-repair-AP}
    In this section, we revisit the framework established in \cite{con2021nonlinear} that utilizes arithmetic progressions to construct repair schemes. 
    For simplicity, we assume the failed node is the $n$th one in this discussion, though this assumption is relaxed in our formal treatment.

A repair scheme for the $n$th node consists of $n-1$ functions $\tau_i: \Fp \rightarrow [s]$ and a function $G:[s]^{n-1} \rightarrow \Fp$. These functions satisfy the following condition for any codeword $(c_1,\ldots,c_n)\in \cC$:
\begin{equation}
\label{eq:repair-func}
    G(\tau_{1}(c_1),\ldots,\tau_{n-1}(c_{n-1}))=c_n.
\end{equation}
In the event of a failure at the $n$th node, each $i$th node (holding the symbol $c_i$) computes $\tau_i(c_i)$ and transmits it, requiring $\lceil \log s\rceil $ bits. The repair process is completed upon receiving the $n-1$ messages $\tau_i(c_i)$, with the $n$th symbol reconstructed using  \eqref{eq:repair-func}. The total bandwidth of the repair scheme, defined as the aggregate number of bits transmitted during the repair, amounts to $(n-1)\lceil \log(s) \rceil$ bits.
	This framework is applicable to any repair scheme, with the challenge primarily lying in the identification of suitable functions $\tau_i$. These functions must be sufficiently informative to allow the collective computation of the symbol $c_n$ from the information about $c_i$. However, they should not be excessively informative to ensure that the size of the image, denoted by $s$, remains small. This balance is crucial for minimizing the total bandwidth used in the repair process.
 
	
To define the functions $\tau_i$, we employ partitions in the following manner. It is important to note that each function $\tau_i$ induces a partition $\{\tau_i^{-1}(a): a \in [s]\}$ of $\Fp$. Conversely, any partition of $\Fp$ into $s$ subsets can define a function, where the function value at a point $a \in \Fp$ corresponds to the index of the subset containing $a$. Therefore, in our subsequent discussion, we will specify the functions $\tau_i$ by constructing partitions of $\Fp$ into $s$ subsets, utilizing arithmetic progressions.

The work in \cite{con2021nonlinear} demonstrates the use of arithmetic progressions for creating partitions that lead to efficient nonlinear repair schemes for RS codes over prime fields. The specifics of this construction are detailed next.

Fix an integer $1 \leq t \leq p$, set $s = \lceil p/t \rceil$, and define $A_0, \ldots, A_{s-1}$ as the partition of $\Fp$ into $s$ arithmetic progressions of length $t$ and step $1$:
\begin{equation} 
\label{eq:Fp-partition}
A_j = \begin{cases}
\{ jt, jt + 1, \ldots, jt + t - 1\} & \text{for } 0 \leq j \leq s-2, \\
\{(s-1)t, \ldots, p - 1\} & \text{for } j = s-1. 
\end{cases}
\end{equation}
For a nonzero $\gamma \in \Fp$, it can be verified that $\gamma \cdot A_0, \ldots, \gamma \cdot A_{s-1}$ also forms a partition of $\Fp$ into arithmetic progressions of length $t$ (except for the last set $\gamma \cdot A_{s-1}$), with step $\gamma$. Each function $\tau_i$, for $i \in [n-1]$, in the repair scheme is defined by a partition $\gamma_i \cdot A_0, \ldots, \gamma_i \cdot A_{s-1}$, with an appropriately chosen $\gamma_i$. It is important to note that distinct $i$'s will have distinct $\gamma_i$'s, ensuring the uniqueness of each $\tau_i$.
As observed in \cite{con2021nonlinear}, these partitions defined by the $\gamma_i$'s lead to a \emph{valid repair scheme} if and only if for any two codewords $c, c' \in \cC$ that belong to the same set in all of the $n-1$ different partitions (i.e., $c_i, c'_i \in \gamma_i \cdot A_{j_i}$ for all $i \in [n-1]$), it holds that $c_n = c'_n$. 

In \cite{con2021nonlinear}, the authors provided a simple yet sufficient condition for a linear code to have a valid repair scheme. We adapt their condition specifically for RS codes, which is the primary focus of this paper.

\begin{prop}\cite[Proposition 2.2]{con2021nonlinear}
\label{con:repair-condition}
    Consider an  $[n, k]_p$ RS code defined with the evaluation points $\alpha_1, \ldots, \alpha_n$. Let $\ell\in[n]$ be the index of the failed node. Let $t<p$ be an integer and for $i\in [n]\setminus\{\ell\}$, let $\gamma_i \in \Fp^*$. 
    If for any polynomial $f(x)\in \Fp[x]$ of degree less than $k$ with  $f(\alpha_i)\in \gamma_i\cdot [-t,t]$ for all $i\in [n] \setminus \{ \ell \}$, it holds that $f(\alpha_{\ell})=0$, then, the $\gamma_i$'s define  a valid repair scheme for the $\ell$th node with a total bandwidth of $(n-1) \cdot \log \ceil{p/t}$ bits.
\end{prop}
In \cite{con2021nonlinear}, the authors concentrated on the regime where $t = \Theta \left(p^{1 - \frac{1}{n-k +1}}\right)$, with $n$ and $k$ being constants relative to $p$. In this setting, each surviving node transmits $\Theta (\log(p))$ bits. Since $p$ is considered to be growing in these results, the number of bits transmitted is not constant.

In contrast, our work explores the regime where $t = \Theta (p)$. In this scenario, every node sends a \emph{constant} number of bits, specifically $\log \lceil p/t \rceil$ bits, to the replacement node. This approach significantly differs from the previous work by reducing the transmission to a constant bit size, regardless of the growth of $p$.

\section{Repair schemes using bounds on Kloosterman sums}
In this section, we present our scheme for repairing single failed nodes for RS codes over prime fields.
We shall extensively use the following simple lemma.
\begin{lem} \label{lem:large-sum} 
    Let $a_1,\ldots,a_n,t \in \Fp$ such that $a_i \in [-t, t]$ for all $i\in[n]$. Then, 
    $
    \left| \sum_{i=1}^{n} \ep{a_i} \right| \geq n\cdot \cos\left(\frac{2\pi t}{p}\right) \;.
    $
\end{lem}
\begin{proof}
    Let $\Re{x}$ denote the real part of the complex number $x$, then
    \begin{align*}
       \left| \sum_{i=1}^{n} \ep{a_i} \right| \geq \Re{\sum_{i=1}^{n} \ep{a_i}} &= \sum_{i=1}^{n}\Re{\ep{a_i}} \geq  n\cdot \cos\left(\frac{2\pi t}{p}\right),
    \end{align*}
    where the second inequality follows since for all $i\in[n]$, we have $a_i \in [-t,t]$.
\end{proof}

The next theorem due to Korolev establishes a nontrivial upper on short Kloosterman sums.
\begin{thm} \cite[Theorems 2, 3]{korolev2016karatsuba} \label{thm:korolev1}
    Let $p$ be a prime and let $n$ be an integer such that
    \[
    e^{(\ln p)^{2/3} \cdot (\ln \ln p)^{1/3}} \leq n \leq \sqrt{p} \;.
    \]
    Let $D = \frac{(\ln n)^{3/2}}{(\ln p) (\ln \ln p)^2}$. Then, it holds
    \begin{enumerate}
        \item For any $a$ such that $\gcd (a,p)=1$,
        \[ 
        \left|\sum_{1\leq \nu\leq n} \ep{\frac{a}{\nu}}\right| \leq  n \cdot\frac{260 \ln D}{D} = o(n)\;.
        \] 
        \item For any $a,b$ such that $\gcd (ab,p)=1$
        \[
    \left|\sum_{1\leq \nu\leq n} \ep{\frac{a}{\nu} + b\nu} \right| \leq  n \cdot\frac{222 \ln D}{D^{3/4}} = o(n) \;.
        \]
    \end{enumerate}
\end{thm}

Next, we present our first repair scheme that is capable of repairing every node in an $[n,3]_p$ RS code for $n$ that is small compared to $p$. 

\label{sec:repair-kls}
\begin{thm}\label{thm:shortKloostermanRepair}
Let $p$ be a large enough prime and $n$ be an integer such that $$2\exp ((\ln p)^{2/3} \cdot (\ln \ln p)^{1/3})\leq n \leq \sqrt{p}\;.$$ Then, the $[n+1, 3]_p$ RS code defined by the evaluation points $\alpha_i=i$ for  $i \in [0, n]$ admits a repair of any node by downloading three bits from all the other nodes.
\begin{proof}
Set $t:= \ceil{p/8}$ and assume that we wish to repair the $\ell$th node. We will prove that the condition in \Cref{con:repair-condition} holds with $\gamma_i = i - {\ell}$ for every $ i \in [0, n] \setminus \{ \ell \}$. Assume towards a contradiction that it does not hold. Thus, there exists a polynomial $f(x) = f_2(x-{\ell})^2 + f_1(x-{\ell}) + f_0 \in \Fp[x]$ such that $f(i) \in (i - {\ell}) \cdot [-t, t]$ for all $i\in [0, n] \setminus \{ \ell \}$, but $f({\ell}) \neq 0$. Define, 
\[
S:=\left| \sum_{\substack{i=0 \\ i\neq \ell}}^{n} \ep{\frac{f(i)}{i - {\ell}}} \right| 
\]
and note that by \Cref{lem:large-sum} and our choice of $t$, it holds that $S \geq n\cdot \cos \left( \frac{2\pi t}{p}\right) > 0.7n$.
On the other hand,
\begin{align*}
S&= \left| \sum_{\substack{i=0 \\ i\neq \ell}}^{n} \ep{f_2 \cdot (i - {\ell}) + f_1 + \frac{f_0}{i - {\ell}}} \right| \\
&= \left| \sum_{\substack{i=0 \\ i\neq \ell}}^{n} \ep{f_2 \cdot (i - {\ell}) + \frac{f_0}{i - {\ell}}} \right| \\
&= \left| \sum_{i=1}^{\ell} \ep{ - f_2 \cdot i - \frac{f_0}{i}} + \sum_{i=1}^{n - \ell} \ep{ f_2 \cdot i + \frac{f_0}{i}}\right|\\
&\leq \left| \sum_{i=1}^{\ell} \ep{f_2 \cdot i + \frac{f_0}{i}} \right| + \left| \sum_{i=1}^{n - \ell} \ep{ f_2 \cdot i + \frac{f_0}{i}}\right|. 
\end{align*}
Clearly, $\max(\ell, n-\ell) \geq n/2$. Assume without loss of generality that $\ell \geq n/2$ and let $S'=  \left| \sum_{i=1}^{\ell} \ep{f_2 \cdot i + \frac{f_0}{i}} \right|$.

Recall that by our assumption  $f(\ell)= f_0 \neq 0$, which implies that $\gcd(f_0, p) = 1$. Consider the following two options. First, if $f_2\neq 0$ then $\gcd(f_0f_2, p) = 1$ and by the second bound in \Cref{thm:korolev1}  $S' = o(n)$. Second, if $f_2 = 0$ then by the first bound in \Cref{thm:korolev1} $S' = o(n)$. In either case, we conclude that $S\leq 0.5n + o(n)$ and we arrive at a contradiction for large enough $p$, which implies a large enough $n$.

Lastly, the claim about the bandwidth follows by noting that by \Cref{con:repair-condition}, each node sends
$
\ceil{ \log \left( \ceil{\frac{p}{t}}\right)} \leq \ceil{\log \left( 8 \right)} = 3
$
bits.
\end{proof}
\end{thm}

\begin{remark}
From a cryptographic perspective, Theorem \ref{thm:shortKloostermanRepair} provides an explicit leakage attack on the secret in a $3$-out-of-$n$ Shamir's Secret Sharing Scheme (SSS) over prime fields. This attack is significant as it requires leaking as little as  $3$ bits from each share to completely reveal the secret. Specifically, we construct deterministic leakage functions that enable the complete recovery of the secret value, $f(0)$. To the best of our knowledge, the only other known attacks on Shamir's SSS are probabilistic in nature, where the adversary randomly chooses the leakage functions \cite{nielsen2020lower}. The above result therefore shows a  significant vulnerability of Shamir's SSS to targeted attacks.\end{remark}

\begin{remark}

    In this section, we construct a repair scheme that downloads a constant number of bits from any node. The total bandwidth is larger than $k\cdot \log(p)$, the amount of information needed to learn the polynomial. One might ask, does the information received from the nodes tell us the entire polynomial?
    We would like to stress that this is not the case here. We show that there are two distinct polynomials on which the scheme transmits the same information. Consider the polynomials $f(x) = x$ and $g(x) = 2x$ and assume that we want to repair the node $0$. A node that stores $f(\alpha)$ or $g(\alpha)$ will transmits the same value, since $f(\alpha), g(\alpha) \in \alpha \cdot [0, t-1]$ for every $\alpha\in \Fp^*$. Thus, we cannot distinguish between these two polynomials upon receiving the information from the nodes.
\end{remark}

\section{Repair scheme using the Weil bound} \label{sec:recover-poly}
In \Cref{sec:repair-kls}, we developed a repair scheme leveraging bounds on Kloosterman sums. This section introduces a new scheme for full-length RS codes, utilizing the Weil bound. This scheme can also repair any node by downloading a constant number of bits from the majority of the other functional nodes. Remarkably, it facilitates the recovery of the entire polynomial (codeword), and not just an individual node.

To recover the entire polynomial, it is necessary to download at least $k \log(p)$ bits in total from the nodes. A straightforward approach would be to retrieve the symbols of any $k$ nodes and then solve the classical interpolation problem to reconstruct the polynomial. Our proposed scheme, however, involves downloading a constant number of bits from more than $k$ nodes. We demonstrate that this information is sufficient to deduce the entire polynomial. 
    
    In order to show that the scheme indeed recovers the entire polynomial, we will first introduce some definitions and later present a stronger condition than the one presented in \Cref{con:repair-condition}. 
    
A \emph{decoding scheme} for a code $\mathcal{C}$ consists of $n$ functions $\tau_i: \mathbb{F}_p \rightarrow A$, where $A$ is some set, and a function $G: A^n \rightarrow \mathbb{F}_p^n$. This scheme must satisfy the condition that for any codeword $(c_1, \ldots, c_n) \in \mathcal{C}$, 
\begin{equation}
\label{stamstam2}
    G(\tau_1(c_1), \ldots, \tau_n(c_n)) = (c_1, \ldots, c_n).
\end{equation}

The following observation provides a condition for the existence of a decoding scheme.
\begin{obs}
\label{obs1}
    Let $\mathcal{C} \subseteq \mathbb{F}_p^n$ be a  code. A set of $n$ functions $\tau_i: \mathbb{F}_p \rightarrow A$ define a valid decoding scheme,  (i.e., there exists a function $G$ satisfying \eqref{stamstam2}) if and only if the mapping 
    \[
    c \in \mathcal{C} \mapsto (\tau_1(c_1), \ldots, \tau_n(c_n))
    \]
    is injective.
\end{obs}
To present the trivial approach suggested above to retrieve the polynomial as a decoding scheme of an $[n,k]$ RS code, let $A=\Fp$ and set, say the first $k$ functions, $\tau_1, \ldots, \tau_k: \mathbb{F}_p \rightarrow \mathbb{F}_p$, as identity functions, and the remaining $n-k$ functions, $\tau_i$, as constant zero functions.

Our goal is to construct decoding schemes non-explicitly and explicitly from a set of $n$ functions where each function outputs only a constant number of bits. We begin with a non-explicit decoding scheme derived by a simple application of the probabilistic method, the proof of which is relegated to Appendix \ref{repair-existence}. We then proceed with the explicit scheme derived from the Weil bound.

\begin{restatable}{prop}{mytheorem}
 \label{clm:decode-exist}
    Consider an $[n,k]$ RS code over $\Fp$ with  $k < \frac{(n + 1)\log(s)}{2\log(p) + \log(s)}$, for some integer $s$. Assume further that the code is defined by the evaluation points $\alpha_1, \ldots, \alpha_n \in \Fp$.
    Then, there are $n$ functions $\tau_1, \ldots, \tau_n:\Fp \rightarrow [s]$ such that for any $f\in \Fp [x]$ of degree less than  $k$, the mapping  $f\mapsto(\tau_1(f(\alpha_1)), \ldots, \tau_n(f(\alpha_n)))$ is injective and therefore it is  sufficient to retrieve $f$ from it. 
\end{restatable}

\Cref{clm:decode-exist} shows that for a constant $s$, any $[n,k]_p$ RS code with $k = O(n/\log(p))$ has a decoding scheme where the functions $\tau_i$ have a constant-size image. This means it is feasible to recover (decode) any polynomial of degree less than $k$ by transmitting a constant number of bits from each node.
However, the result is existential and not explicit. 
 On the other hand, the following simple claim shows that for $k$ which is slightly larger than the bound in \Cref{clm:decode-exist}, a decoding scheme does not exist for an $[n,k]_p$ RS code.
\begin{claim}
    Consider an $[n,k]$ RS code over $\Fp$ with $k>\frac{n\log s}{\log p}$. Then, for any $n$ functions $\tau_1, \ldots, \tau_n: \Fp \rightarrow [s]$, the mapping $c\in \cC \rightarrow (\tau_1(c_1), \ldots, \tau_n(c_n))$ is not injective.
\end{claim}
\begin{proof}
    There are $p^k$ codewords and $n^s$ vectors $[s]^n$. Thus, for $k>\frac{n\log s}{\log p}$, by the pigeonhole principle, there are two codewords that map to the same element under $\tau_1,\ldots,\tau_n$.
\end{proof}
By combining the above results, we conclude that the maximum dimension of RS codes that permits a decoding scheme with a constant number of bits from each node is \( k = \Theta(n/\log p) \). Our next result, while falling short of this benchmark, presents an explicit decoding scheme. Specifically, we introduce an explicit decoding scheme for a full-length RS code, i.e., of length \( p \) with \( k = O(\sqrt{p}) \). This dimension is smaller than the maximum possible dimension which is \( \Omega(p/\log p) \) for a full-length RS code. An intriguing open question is the construction of an explicit decoding scheme for a full-length RS code with dimension \( \Omega(p/\log p) \) that requires only a constant number of bits from each node.

As in the case of the repair schemes, we construct the function $\tau_i$ by partitioning $\mathbb{F}_p$ into arithmetic progressions with step $\gamma_i \in \mathbb{F}_p$. The following proposition provides a condition similar to that in \Cref{con:repair-condition} for the functions $\tau_i$ defined by the $\gamma_i$'s to define a valid decoding scheme.
Furthermore, we allow also that some of the nodes do not send any information; equivalently, the function they compute on their data is a constant function.
\begin{prop} \label{con:decode-cond}
    Consider an $[n,k]_p$ RS code defined by the evaluation points $\alpha_1, \ldots, \alpha_n$. Let $t < p$ be an integer, $M\subset [n]$ of size $m$, and $\gamma_i, i \in [n]\backslash M$, be nonzero elements of $\mathbb{F}_p$. If the only polynomial $f \in \mathbb{F}_p[x]$ of degree less than $k$ for which $f(\alpha_i) \in \gamma_i [-t,t]$ for all $i \in [n]\backslash M$ is the zero polynomial, then the $\gamma_i$'s define a valid decoding scheme with a total bandwidth of $(n-m) \cdot \log \left( \lceil \frac{p}{t} \rceil \right)$ bits.
\end{prop}
Note that in the scheme in \Cref{con:decode-cond}, the nodes with indices $i\in M$ are the nodes that do not send any information.
\begin{proof}
    Let $s = \lceil \frac{p}{t} \rceil$ and define the sets $A_j$, $0 \leq j \leq s-1$, as in \eqref{eq:Fp-partition}. For each $i \in [n]\backslash M$, define the function $\tau_i$ according to the partition $\gamma_i \cdot A_0, \ldots, \gamma_i \cdot A_{s-1}$, i.e., $\tau_i(a) = j$ if and only if $a \in \gamma_i \cdot A_j$.
    
    By Observation \ref{obs1}, it suffices to show that the mapping for $c \in \mathcal{C}$, $c \mapsto (\tau_1(c_1), \ldots, \tau_n(c_n))$, is injective. Let $c, c' \in \mathcal{C}$ be two codewords that agree on the $n$ values $\tau_i(c) = \tau_i(c')$ for all $i \in [n]\backslash M$. Then, $c - c'$ is a codeword such that its $i$th symbol belongs to the set $\gamma_i \cdot [-t,t]$ for all $i \in [n]$. Therefore, by the condition, $c - c'$ must be the zero codeword, i.e., $c = c'$. The claim about the bandwidth follows since each partition consists of exactly $s$ sets.
\end{proof}

Next, we introduce our decoding scheme that utilizes the Weil bound. The scheme facilitates decoding even in scenarios where a subset of the nodes does not transmit any information. 

\begin{thm} \label{thm:Weil-recover-const}
    Let $B\geq 3$ be a positive integer. Let $k$ and $m$ be positive integers and $p$ be a prime such that $k + m \leq  \cos (\frac{2\pi}{2^B}+\frac{2\pi}{p})\sqrt{p}$. 
    Then, the full-length $[p, k]_p$ RS code admits a decoding scheme by downloading $B$ bits from any  $p-m$  nodes.
\end{thm}
Before delving into the proof of the theorem, it is important to note the impact of increasing $B$, the number of bits transmitted by the nodes. By transmitting more information, we can increase the sum of $k + m$. This implies that the decoding scheme can be adapted either to RS codes with a larger dimension or to scenarios where a greater number of nodes do not transmit any information.
\begin{proof}
    Set $t:= \ceil{p/2^B}$ and assume that the nodes $\ell_1, \ldots, \ell_m\in \Fp$ do not transmit any information during the scheme. 
The result will follow by showing that the condition in \Cref{con:decode-cond} holds with  $\gamma_i = \prod_{j=1}^m (i - \ell_j)^{-1}$ for all $i\in \Fp \setminus \{\ell_1, \ldots, \ell_m \}$. Assume towards a contradiction that it does not hold. Then, there exists a nonzero polynomial $f(x) = \sum_{j=0}^{k-1} a_jx^{j}$ such that $f(i)\in \left( \prod_{j=1}^m (i - \ell_j)^{-1} \right) [-t, t]$ for every $i\in \Fp \setminus \{\ell_1, \ldots, \ell_m  \}$.
    Thus, by \Cref{lem:large-sum}
    \[
    \left| \sum_{i\in \Fp} \ep{f(i) \cdot \prod_{j=1}^m (i - \ell_j)} \right| \geq p\cos(2\pi t/p)>
    p\cos(\frac{2\pi}{2^B}+\frac{2\pi}{p})
    \;.
    \]
    where the last inequality follows $2\pi t/p \leq \frac{2\pi}{2^B}+\frac{2\pi}{p} < \pi$.  
    On the other hand, by assumption $f\not \equiv 0$, and thus $f(x) \cdot \prod_{j=1}^m (x - \ell_j)$ is a non-constant polynomial of degree at most $k + m - 1$ and therefore by the Weil bound (Theorem \ref{thm:weil-bound})  
    \[
    \left| \sum_{i\in \Fp} \ep{f(i) \cdot \prod_{j=1}^m (i - \ell_j)} \right|\leq (k + m - 1)\sqrt{p},
    \]
    and we arrive at a contradiction by the choice of $k$ and $m$.
    The claim on the bandwidth follows from \Cref{con:decode-cond}.
\end{proof}


\begin{remark}
    We note that the decoding scheme presented in \Cref{thm:Weil-recover-const} is in particular, also a repair scheme. Indeed, any node of the $m$ nodes that does not send any information can be viewed as the failed node.
\end{remark}
\begin{remark}
    In \Cref{thm:Weil-recover-const}, we utilized the classical Weil bound on exponential sums, which provides an upper bound for \emph{complete} sums, i.e., sums that range over all elements in the prime field. However, it is important to note that there are also upper bounds for \emph{incomplete sums} (see, for example, \cite{ostafe2022weil} that bounds sums ranging over small subgroups). These bounds on incomplete sums, can be employed to construct decoding schemes for RS.
\end{remark}

Recall that \Cref{thm:secret-sharing-leakage-cons-rate} implies that repairing the \(0\)th node in an RS code is not feasible when the rate is sufficiently high. Specifically, for any positive integer \(m\), there exists a constant \(\alpha_m\) such that an \([p, \alpha_m p]_p\) RS code does not admit a repair of the \(0\)th node using \(m\) bits of information from each of the other nodes.

On the other hand, \Cref{thm:Weil-recover-const} allows for the repair of the \(0\)th node in a full-length RS code with dimension \(O(\sqrt{p})\). An interesting question arises: how far can one push the dimension of the code while using this scheme? In other words, what is the largest dimension of a full-length RS code that permits decoding (or repairing) via the scheme given in \Cref{thm:Weil-recover-const}? The following example illustrates some of the limitations of this scheme.

\begin{example}Consider the \([p, \frac{p-1}{2}]_p\) RS code over \(\mathbb{F}_p\). Our goal is to show that the decoding scheme given in \Cref{thm:Weil-recover-const} is not valid, even if all of the nodes except the \(0\)th send almost all their information.    
Consider the codeword defined by the polynomial 
\[f(x) = -\sum_{i=0}^{\frac{p-3}{2}} (x+1)^i = \frac{1 - (x+1)^{\frac{p-1}{2}}}{x},\]
of degree \(\frac{p-3}{2}\). 
Since we assume that all the nodes send their information except for the \(0\)th node, then by the definition of \(\gamma_i\) in the proof of \Cref{thm:Weil-recover-const}, we have that 
\(\gamma_i=1/i\) for each \(i \in \mathbb{F}_p^*\). Next, let \(t=3\) and define the function \(\tau_i\) as in \Cref{sec:framework-repair-AP}. Since \(t=3\), the information \(\tau_i(f(i))\) received from the \(i\)th node suffices to determine \(f(i)\) up to \(3\) possible values. Furthermore, \(\tau_i\) is defined by the partition \(\gamma_iA_j, j=0,\ldots, \lfloor p/3 \rfloor\) of \(\mathbb{F}_p\) as in \eqref{eq:Fp-partition}. 
Since \(f(i) \in i^{-1} \cdot \{0,1,2\}\) for every \(i \in \mathbb{F}_p^*\), we have \(\tau_i(f(i))=0\). On the other hand, for \(g(x) \equiv 0\), also \(\tau_i(g(i))=0\) for all \(i \in \mathbb{F}_p^*\) and \(g(0)=0 \neq -\frac{p-1}{2} = f(0)\). We conclude that the received information does not suffice to recover the value of \(f\) at \(0\), let alone decode \(f\).
%
%
%
%
%
\end{example}
\section{Summary and Open Problems} \label{sec:summary}
This paper extends the work initiated in \cite{con2021nonlinear}, focusing on the repair of RS codes over prime fields. We advance this line of work by developing nonlinear repair schemes for RS codes, where each node transmits only a constant number of bits. While the total bandwidth of our schemes exceeds that of the trivial repair scheme, they offer a significant advantage in scenarios where each node is limited in power or communication and is thus constrained to transmit only a constant number of bits—a condition under which the trivial scheme is not feasible.

Furthermore, we demonstrate that our schemes can be interpreted as explicit local leakage attacks on Shamir's SSS. To our knowledge, these are the first instances of such explicit attacks.

The main contribution of this work is establishing connections between bounds on exponential sums and the problem of repairing RS codes. By applying some known bounds, we obtain results that imply both a repair and a decoding scheme for certain RS codes.

We conclude with open questions for future research:
\begin{enumerate}
    \item Develop efficient algorithms to compute the repair function \(G\) as defined in \eqref{eq:repair-func}, i.e., efficiently compute the lost data upon receiving inputs from the functioning nodes.
    \item Determine the maximal rate of an RS code over prime fields for which there is a repair scheme that downloads a constant number of bits from each node. Constructing such schemes remains an interesting open question.
\end{enumerate}

\bibliographystyle{plain}
\bibliography{refs}
\appendix
\section{Appendix}
\subsection{Proof of \Cref{clm:decode-exist}}
\label{repair-existence}
For the reader's convenience we restate the proposition.
\mytheorem*

Before delving into the proof of \Cref{clm:decode-exist}, it is important to note its resemblance to the proof of Theorem 2 in \cite{nielsen2020lower}. Both proofs employ a similar strategy: they randomly select the functions $\tau_i$ and then apply the union bound. However, a key distinction lies in the objective of the recovery process. While our approach aims to recover the entire polynomial, the focus of Theorem 2 in \cite{nielsen2020lower} is solely on recovering the secret, i.e., the value of the polynomial at zero.
\begin{proof}
    Pick the functions $\tau_i: \Fp \rightarrow{[s]}$  uniformly at random. Namely, for each $i\in [n]$ and  $\alpha \in \Fp$,  $\tau_i(\alpha)$ is a uniformly distributed random element of $[s]$. Fix two distinct polynomials $f,g$ of degree less than  $k$, then it holds that 
    \begin{align*}
        \Pr_{\tau_1, \ldots, \tau_n} \left[(\tau_1(f(\alpha_1)), \ldots, \tau_n(f(\alpha_n))) = (\tau_1(g(\alpha_1)), \ldots, \tau_n(g(\alpha_n)))\right] &= \prod_{i=1}^n \Pr \left[ \tau_i(f(\alpha_i)) = \tau_i(g(\alpha_i)) \right] \;.
    \end{align*}
    If $f(\alpha_i) = g(\alpha_i)$, then necessarily $\Pr \left[ \tau_i(f(\alpha_i)) = \tau_i(g(\alpha_i)) \right] = 1$, and since $f$ and $g$ are polynomials of degree less than  $k$ they agree on at most $k-1$ of the $\alpha_i$'s. If $f(\alpha_i) \neq g(\alpha_i)$, then by the definition of $\tau_i$, we have that $\Pr \left[ \tau_i(f(\alpha_i)) = \tau_i(g(\alpha_i)) \right] = 1/s$. Thus, 
    \[
    \Pr_{\tau_1, \ldots, \tau_n} \left[ (\tau_1(f(\alpha_1)), \ldots, \tau_n(f(\alpha_n))) = (\tau_1(g(\alpha_1)), \ldots, \tau_n(g(\alpha_n))) \right] \leq \left( \frac{1}{s} \right)^{n-k+1}\;.
    \]
    Now, by the union bound, the probability that there exist two polynomials $f$ and $g$ of degree less than $k$ such that $(\tau_1(f(\alpha_1)), \ldots, \tau_n(f(\alpha_n))) = (\tau_1(g(\alpha_1)), \ldots, \tau_n(g(\alpha_n)))$ is at most $p^{2k}\cdot (1/s)^{n-k+1}<1$ by the choice of $k$. Thus with positive probablity the required functions $\tau_i$ exist, and the result follows.
%
%
\end{proof}

\end{document}